\documentclass{egpubl}
\usepackage{egsgp2025}

\SpecialIssuePaper         
\CGFStandardLicense

\usepackage[T1]{fontenc}
\usepackage{dfadobe}  

\usepackage{cite}  %
\BibtexOrBiblatex
\electronicVersion
\PrintedOrElectronic
\ifpdf \usepackage[pdftex]{graphicx} \pdfcompresslevel=9
\else \usepackage[dvips]{graphicx} \fi

\usepackage{egweblnk}

\usepackage{layouts}
\usepackage{wrapfig}
\usepackage{amsmath}
\usepackage{amssymb}
\usepackage{booktabs}
\usepackage{algorithm}
\usepackage[noend]{algpseudocode}

\def\etal{et~al.}
\def\eg{e.g.}               %

\newtheorem{theorem}{Theorem}

\title[Uniform Sampling of Surfaces by Casting Rays]%
{
Uniform Sampling of Surfaces by Casting Rays
}

\author[S. Ling \& A. Madan \& N. Sharp \& A. Jacobson]
{\parbox{\textwidth}{\centering Selena Ling$^{1}$, Abhishek Madan$^{1}$, Nicholas Sharp$^{3}$
       and Alec Jacobson$^{1,2}$
       }
       \\
{\parbox{\textwidth}{\centering $^1$University of Toronto, Canada\\
        $^2$Adobe Research, Canada\quad
        $^3$NVIDIA, USA
       }
}
}

\begin{document}

\teaser{
\vspace*{-2.8em} 
 \includegraphics[width=\linewidth]{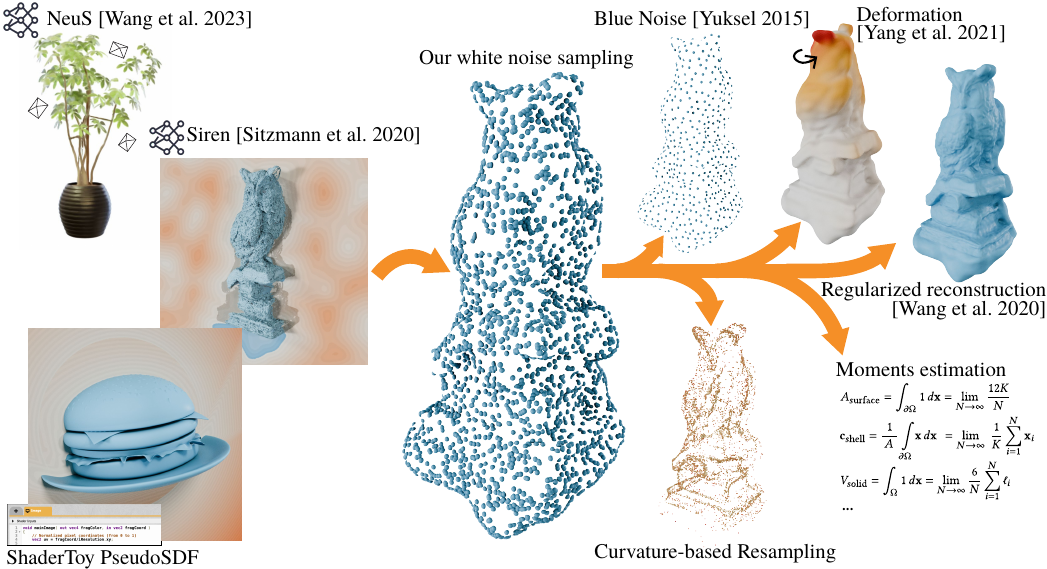}
 \centering
  \caption{
  Implicit surfaces represent the geometry of a shape as the zero level set of a function of 3D space. There are many varieties of implicit surfaces used in practice (left).
  We study a method for a fundamental surface operation: uniformly random point sampling.
  Any point on the continuous surface is equally likely to be chosen through our proposed process of sampling random rays in space and finding intersections with the implicit (typically with sphere tracing).
  This white noise sampling can then act as the \emph{raw ingredient} for downstream recipes, including blue noise sampling, curvature-based importance resampling, deforming neural implicits, sampling-based regularization terms for reconstruction, and direct estimation of basic shape quantities (surface area, enclosed volume, center of mass, etc.). ShaderToy PseudoSDF credit to \copyright Xor \href{https://creativecommons.org/licenses/by-nc-sa/3.0/}{(CC BY-NC-SA 3.0)}}
  }
\label{fig:teaser}

\maketitle
\begin{abstract}
    Randomly sampling points on surfaces is an essential operation in geometry processing. This sampling is computationally straightforward on explicit meshes, but it is much more difficult on other shape representations, such as widely-used implicit surfaces.  This work studies a simple and general scheme for sampling points on a surface, which is derived from a connection to the intersections of random rays with the surface. Concretely, given a subroutine to cast a ray against a surface and find all intersections, we can use that subroutine to uniformly sample white noise points on the surface. This approach is particularly effective in the context of implicit signed distance functions, where sphere marching allows us to efficiently cast rays and sample points, without needing to extract an intermediate mesh. We analyze the basic method to show that it guarantees uniformity, and find experimentally that it is significantly more efficient than alternative strategies on a variety of representations. Furthermore, we show extensions to blue noise sampling and stratified sampling, and applications to deform neural implicit surfaces as well as moment estimation.

\end{abstract}  

\section{Introduction}

Sampling uniformly distributed points on surfaces is essential for characterizing the underlying geometry in many downstream applications such as visualization and simulation.
Concrete algorithmic use-cases include integrating surface areas or other geometric properties~\cite{li2003using}, evaluating regularizers for geometric optimization~\cite{isopoints}, computing metrics such as Chamfer distance~\cite{diffcd}, and sampling BSSRDF exit points in rendering~\cite{king2013bssrdf}.
It is easy to draw samples on a surface explicitly defined as a mesh, however it is not so straightforward when no mesh is available, such as implicit surfaces. 
This is problematic, as implicit surfaces are increasingly widely used, for instance as a neural representation for shapes in machine learning.
We consider the problem of sampling uniformly-distributed points on surfaces in this more general setting.

How might one sample points on an implicit surface? One possibility is rejection sampling, drawing random points in space and keeping only those near the surface, but this wastes many samples and requires an error-inducing projection onto the surface.
Another is to use isosurface extraction such as marching cubes to recover a mesh, and sample from that, but this requires expensive sampling to a grid and can alias fine features.
Principled sampling processes like Markov chain Monte Carlo guarantee a proper distribution in an asymptotic limit, but convergence may be slow in practice.

Instead, we leverage a classic mathematical relationship to the intersection of random lines with the surface (e.g., ``Cauchy-Crofton Formula'' and ``Buffon's Needle Problem'').
Taking random lines drawn from an appropriate distribution and gathering all intersections yields a uniform sampling of the surface---see \cite{palaispoints} for one reference.
Precisely, this strategy produces a white-noise uniform distribution of samples, although other samplings can also be obtained with our method (Figure~\ref{fig:bluenoise}).
This relationship is well-known in mathematics. It has sporadically and briefly been studied in visual computing~\cite{Detwiler_2008}, but has not previously been put to work in the context of modern implicit surfaces and neural representations, where we show that it offers significant benefits.

Importantly, this sampling method applies to any surface representation for which we have the ability to intersect rays with the surface.
Because ray casting is already a necessary operation for rendering and visualization, it is widely and efficiently implemented for a wide variety of implicit surfaces and other nonstandard shape representations.
For instance, we can efficiently intersect rays with implicit signed distance functions surfaces via sphere tracing~\cite{spheretracing}, with harmonic functions via Harnack tracing~\cite{gillespie2024ray}, or more general implicit functions via interval tracing~\cite{stol1997self}.
Even recent Gaussian particle representations allow for efficient ray tracing~\cite{moenne20243d}.

\begin{figure}[htb]
  \includegraphics[width=\linewidth]{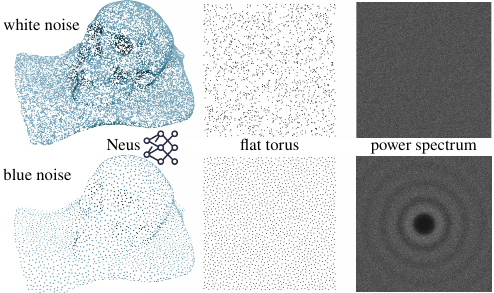}
  \caption{\label{fig:bluenoise}
  White noise sampling is often useful as a raw ingredient to downstream applications, such as blue noise generators (e.g., \cite{bluenoise,bridsonbluenoise}. Our white noise samples on a neural implicit \cite{neus} (top left) are subsampled to a blue noise sample set \cite{bluenoise}. For reference, we show the same process on a periodic square and its corresponding power spectrum images \cite{SchlomerD11}.
  }
\end{figure}

In this paper we study ray intersection-based sampling in the context of implicit surfaces and neural representations. 
We provide self-contained proofs supporting the claim of uniformity and numerically compare with other existing approaches, showing significantly improved efficiency and uniformity. 
We also demonstrate extensions to downstream applications of the approach, including stratified sampling (Section \ref{sec:sparsevoxels}), blue noise sampling (Section \ref{sec:bluenoise}), implicit deformation (Section \ref{sec:deformation}), and more.

\pagebreak

\section{Related Work}

\begin{figure*}[htb]
  \includegraphics[width=\linewidth]{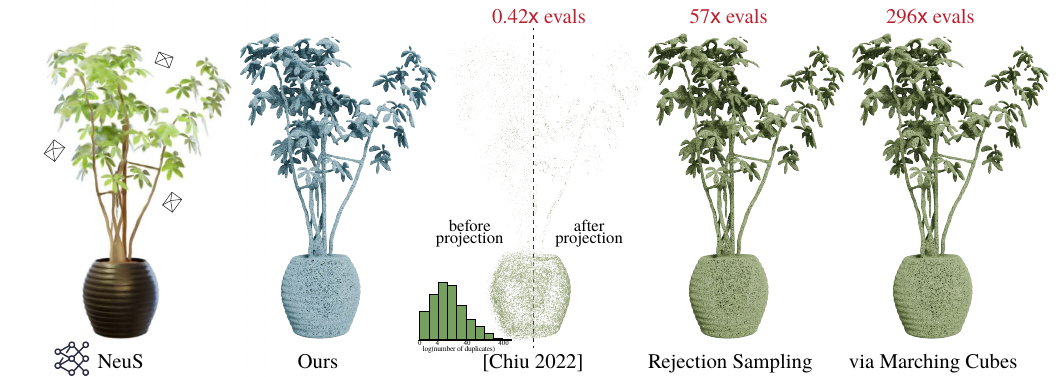}
  \caption{\label{fig:ficus-comparison}
  We visualize 500,000 samples using our method and all baseline methods on the level set of a neural implicit function \cite{neus} learned from images of the ficus scene in the Blender dataset. Our method achieves similar or better sample quality as rejection sampling and sampling via marching cubes on a grid of resolution $4096^3$ while being a factor of 57$\times$ and 296$\times$ cheaper in terms of function evaluations, respectively (for reference, ours took $2.3 \times 10^8$ function evaluations). Because the specialized Hamiltonian Monte Carlo method from \cite{ericauniform} can only sample near the surface (left slice), it requires an additional non-trivial projection to result in on-surface samples (right slice).
  This method also results in duplicated samples (which are necessary for correct statistics), so
  we also plot a histogram of duplicated sample counts.
 Both rejection sampling and sampling via marching cubes also require projection, but is too small to visualize (though it does mildly affect exact uniformity).
  }
\end{figure*}

Sampling surfaces is a core operation used across visual computing from rendering to numerical integration and beyond.
Here we focus specifically on past work for sampling when explicit mesh representations are not available, such as implicit surfaces.

\subsection{Sampling and Implicit Surfaces}
Implicit surfaces are a flexible and general representation, defining surfaces as the level set of functions such as \textit{signed distance functions} or \textit{occupancy functions}. Recent research across visual computing has leveraged neural implicit functions as a shape representation for tasks in learning and reconstruction~\cite{deepsdf, neus, volsdf}.
Across these works, sampling the level set has proven to be helpful for many downstream tasks such as rendering from the underlying shape \cite{volsdf,neus}, improving its optimization by defining on-surface regularization terms \cite{controllinglevelset, isopoints, volsdf, diffcd, neuralindicator} and shape manipulation \cite{gpnf}. 

\subsection{Sampling Algorithms}

Although it is useful to sample points from surfaces which lack an explicit representation, it is not obvious how to do so.
Many strategies have been considered in past work.

\paragraph*{Rejection Sampling}

A basic standard approach is \emph{rejection sampling}: drawing points at random in the domain and keeping only those which lie on the surface, see \eg{} \cite{gpnf,diffcd}.
Rejection sampling is simple and straightforward, but may require a huge number of rejected samples.
Furthermore, because surfaces are co-dimensional sheets with no volume, points will never land exactly on the surface---in practice one must either retain samples in a narrow band around the surface (not truly sampling the surface), or project onto the surface as a post-process (concentrating points in positively-curved regions).

\paragraph*{Grid Based}
Another possibility is to evaluate an implicit function on a regular grid for sampling.
Most commonly, isosurfacing algorithms such as marching cubes~\cite{lorensen1998marching} are used to construct an explicit mesh as an intermediary for sampling~\cite{darmon2022improving,diffcd,neuralindicator}.
Yan \etal{} \cite{yan2014unbiased} apply a related grid-decomposed sampling by observing that the surface can be represented as a height function per-cell, although the function is still linearized within each cell.
These strategies nicely leverage the well-established tools of extraction and mesh sampling, but may suffer from excessive computation and aliasing of fine features from working on a grid.

\paragraph*{Particle Evolution and MCMC}
Other methods iteratively update a sample set to converge to the desired distribution to gradually discover a suitable set of samples.
Wang et al.\cite{isopoints} approximate blue noise samples to assist in training, via a three-stage procedure of projecting, resampling, and upsampling, although they target improved training efficacy more so than any precise sampling distribution.
Such strategies can be formalized via Markov chain Monte Carlo (MCMC), which provably converge to a uniform sample set in the asymptotic limit, as studied in the thesis of Chiu \cite{ericauniform} for implicit surface sampling.
Likewise, Langevin dynamics-based formulations evolve to the desired distribution according to a stochastic differential equation~\cite{cai2020learning,gpnf}, and~\cite{samplingmollified} sample from a mollified interaction energy.
However, these methods may require large numbers of iterations to ``warm up'' to a uniform distribution on complex geometry, and most still require an error-inducing projection step to generate points exactly on the surface.

\paragraph*{Line Intersections}
The sampling method studied in this work follows from a classic relationship of random lines intersecting a surface, see Section~\ref{sec:theory}. 
Although this approach has not yet been utilized in recently important applications with implicit surfaces where it has significant advantages, it has occasionally appeared elsewhere in visual computing, which we outline here.
Detwiler \etal{} use random lines to sample points specifically in the context of CAD geometry~\cite{Detwiler_2008}, although their approach requires specifying an upper bound on the number of intersections \textit{a priori}, which may be prohibitive in practice.
The same approach is adapted to point clouds~\cite{liu2006quasi} via the ray tracing approach of~\cite{schaufler2000ray}, and digital binary-voxel geometry~\cite{liu2010surface}, both for the particular purpose of estimating surface area.
Some other works use low-discrepancy sequences to obtain surface samples on CAD objects~\cite{li2003using} for estimating surface area, and for sampling on meshes~\cite{rovira2005point}.
Beyond surface sampling, ray marching has also been used for sampling on the medial axis~\cite{yeh2014umaprm}.

\begin{figure*}[htb]
  \includegraphics[width=\linewidth]{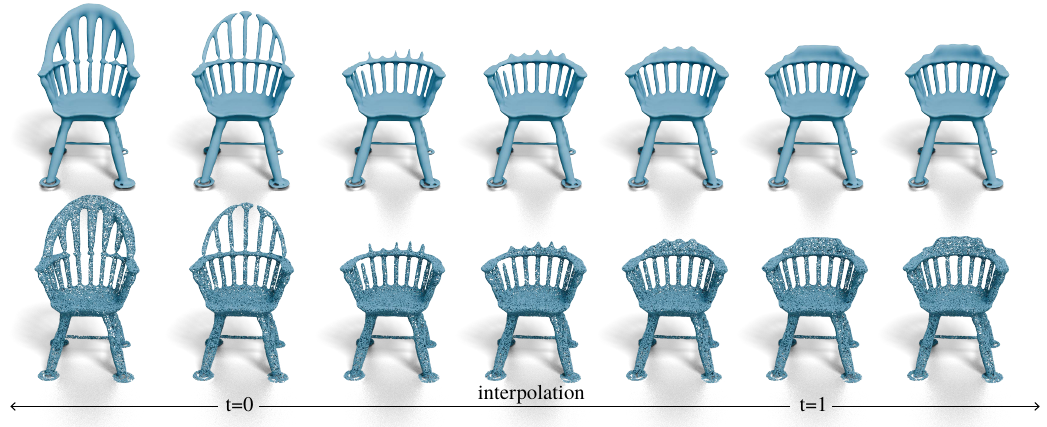}
  \caption{\label{fig:interpolation}
    A neural field is learned to represent a smooth interpolation between two different chairs at $t=0$ and $t=1$ as in ~\cite{liu2022learning} (\emph{top}). 
    Our method uniformly samples the implicit surface of the interpolated shape at any point in the sequence (\emph{bottom}). 
  }
\end{figure*}

\subsection{Other Notions of Sampling}
We primarily consider uniform sampling in the sense of white noise, but other kinds of sampling are also important in visual computing.
Low discrepancy sequences and stratified sampling~\cite{li2003using,rovira2005point,quinn2009low} seek to reduce variance at bounded sample counts by spatially distributing the sample sequence across the domain.
Blue noise sampling seeks a different distribution, with nicely-spaced samples on the surface which are useful for texture synthesis and perceptual optimization ~\cite{bridsonbluenoise, bluenoise, oztireli2010spectral}; in Section~\ref{sec:bluenoise} we show how our sampler can do the same on surfaces.

\section{Theory}
\label{sec:theory}

Our method is derived from the relationship between random lines and their intersections with a surface.
These properties are well-known in integral geometry---the Cauchy-Crofton Formula relates the length of a curve to the number of intersections with a random line, while Buffon's needle problem considers the likelihood of a line segment intersecting parallel strips.
One introduction to sampling via line intersections can be found in~\cite{palaispoints}.
In this section, we provide a self-contained introduction to equidistributed sequences and some intuition for why the algorithm works.

\subsection{Equidistributed Sequences}
A crucial building block of randomized algorithms is an underlying sequence of random numbers that are uniformly distributed, as they provide a simple and flexible primitive for stochastic computation.
Instead of attempting to obtain truly random numbers, pseudo-random number generators (PRNGs) are used in practice, which produce long deterministic sequences with statistical properties that emulate truly random uniformly distributed sequences.
One desirable property is that the proportion of samples in a given region of the sampling domain should be roughly equal to the relative area of that region; sequences that exhibit this property are called \textit{1-equidistributed}~\cite{franklindeterministic}.
An equivalent characterization is that, for any function $f$ over the sampling domain $\Omega$ (an $m$-dimensional manifold embedded in $\mathbb{R}^n$ for any $1 \le m \le n$), the sequence $\{ x_i \}$ can be used to construct a Monte Carlo estimator of $\int_\Omega f(x)dx$ with uniform contribution weights that converges as the sample count $N$ approaches infinity (i.e., $\lim_{N \to \infty} \frac{1}{N} \sum_{i=1}^N |\Omega| f(x_i) = \int_\Omega f(x) dx$) --- such an estimator is called \textit{consistent}.
The notion of equidistribution can be extended to arbitrarily large $k$ by instead considering the convergence of the estimator $\lim_{N \to \infty} \frac{1}{N} \sum_{i=1}^{N-k+1} |\Omega| f_k(x_i, \ldots, x_{i+k-1}) = \int_{\Omega^k} f_k(x_1, \ldots, x_k) dx_1 \ldots dx_k$ for a function of $k$ variables $f_k$.
If a sequence is $k$-equidistributed for all natural numbers $k$, it is \textit{completely equidistributed}.
Since $k$-equidistribution for increasingly large $k$ imposes more and more uniformity requirements on the sequence, it suffices to use $k$-equidistributed sequences for large $k$ in numerical computation, in place of samples from a truly uniform distribution.

\begin{figure}[htb]
  \includegraphics[width=\linewidth]{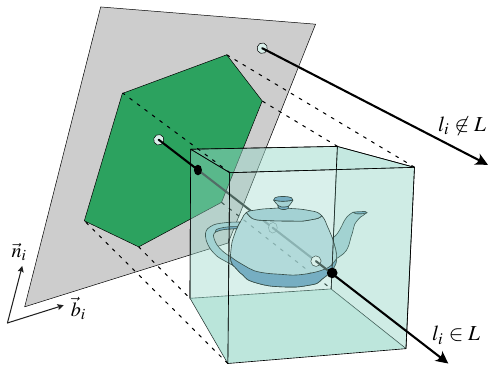}
  \caption{\label{fig:line-sampling}
           An illustration of our uniform ray sampling algorithm as in Algorithm \ref{alg:linesampling}. Given the blue $[-1,1]^3$ cube around an implicitly-defined teapot shape, we first sample a direction $\vec{d}_i$ uniformly. Among all rays parallel to that direction, we only want the rays that intersect the cube to add to our ray sample set $L$, which we achieve by rejection sampling on the gray $2\sqrt{3}$-sided square in the plane spanned by normal $\vec{n}_i$ and bi-normal $\vec{b}_i$. The points whose emanated rays intersect the bounding box fall in the green region on the plane, which is the orthographically projected area of the bounding box onto the plane.
           }
\end{figure}

\subsection{Sampling via Oriented Line Casting}

The surface sampling algorithm is as follows: uniformly sample $M$ oriented lines, and for each line, append \textit{all} of its surface intersections with $\Omega$ into a list of points (this algorithm was previously proposed by Palais \etal~\cite{palaispoints}).
At first glance, it is surprising that such a method truly produces uniform surface samples, but building upon the Cauchy-Crofton formula from integral geometry~\cite{santalointegral}, by taking all surface intersections for each line rather than, e.g., a single intersection per line, the resulting sequence is $k$-equidistributed if the line distribution is also $k$-equidistributed.

Below we provide some intuition for the two key ideas behind the proof for the 1-equidistributed case; the general $k$-equidistributed case follows similarly (and see \cite{palaispoints} for a rigorous proof).
The first key idea is that each point on the surface is intersected by exactly one line along each direction $\vec{d} \in S^2$, so as long as we uniformly sample lines in every direction within a region enclosing the surface, we can ``average out'' the contribution of a point to an arbitrary integral over the surface over all line directions.
The second key idea is that different lines can intersect different numbers of points, which means the lines need non-uniform densities (e.g., lines with more intersections should have a larger density).
Although it is intractable to obtain the true line densities \textit{a priori}, returning every intersection along the line essentially amortizes a line's true non-uniform density and gives each point the desired uniform density, and so the consistent Monte Carlo estimator can simply use uniform weights.
An alternative strategy was used by Detwiler \etal{} \cite{Detwiler_2008} where they rejected lines based on the number of intersections relative to the maximum possible number of intersections with $\Omega$, but in our approach extracts at least one surface sample from every ray and can extract several samples for rays with many intersections, making it more efficient while still maintaining uniformity.

The samples produced by this method are not independent (though they are identically distributed), as the collinear points produced by intersections with the same line are correlated.
However, the definition of $k$-equidistribution ensures that the points are sufficiently well-distributed such that the correlation between samples does not influence downstream applications.

For the rest of the paper, we will refer to \textit{rays} instead of oriented lines; although rays have an associated origin and lines do not, as we will describe in Section~\ref{sec:ray-sampling}, 
\begin{figure}[htb]
  \includegraphics[width=\linewidth]{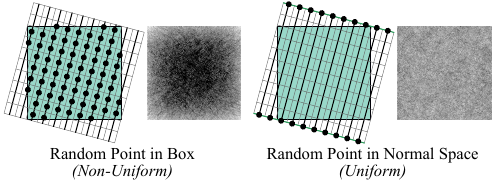}
  \caption{\label{fig:line-offsets}
            Ray origin sampling is crucial for obtaining uniform ray samples. If origins are simply chosen as random points within the box (left), the rays will pass through the center of the box more frequently than the edges (middle left). Meanwhile, if the origins are chosen in the space normal to the line direction (middle right), then the ray density is uniform over the box (right).
           }
\end{figure}
rays are sampled by selecting an origin outside of the bounding cube, such that exactly one ray corresponds to each oriented line (see Figures~\ref{fig:line-sampling},~\ref{fig:line-offsets}).

\subsection{Ray Resampling}
\label{sec:resampling}

The one-to-many sampling procedure described above can also be viewed through the lens of resampling~\cite{rubin1987sir,talbot2005importance}: each ray can be given a weight proportional to the number of their intersections with the surface, and then the rays can be resampled with replacement with probabilities to their weights, after which an intersection along the ray is randomly selected with equal probability.
Taking every single intersection for each ray is then a particular instantiation of the aforementioned resampling process, where each ray $\ell_i$ with $k_i$ intersections is resampled $k_i$ times, and each time a different intersection along the ray is selected.
We formally prove the equivalence that resampling produces uniformly distributed samples in Appendix~\ref{app:resampling}, and empirically show that resampling produces samples of comparable quality in Figure~\ref{fig:resampling}.

\section{Method}

In this paper, we propose a method that uniformly samples the level set of a given implicit function via random ray casting and finding the intersections between the rays and the level set within the bounding domain. The inputs to our algorithm is a bounding box defining the sampling domain and an implicit function with a Lipschitz constant bound $\lambda$. The samples extracted from our method lies strictly on the implicitly-defined surface without needing additional projection and is uniformly distributed on the surface by the definition of uniformity we provided in above section. 

\subsection{Uniform Ray Sampling}
\label{sec:ray-sampling}

Sampling rays uniformly without any bias is an essential first step in our method. However, it is impossible to sample directly in an unbounded manner in the same way that sampling all points in $\mathbb{R}^3$ is impractical. Our method is only interested in rays that will possibly intersect the level set defining the surface, and therefore it is reasonable to restrict the problem of uniform ray sampling to uniformly sampling 
\begin{algorithm}[h!]
\caption{UniformRays: Sampling $n$ uniformly distributed rays within bounding box $[-1,1]^3$.}
\label{alg:linesampling}
\begin{algorithmic}
\Require $M$: number of rays
\Require $bbox$: bounding box of size $[-1,1]^3$ 
\State $L \gets \{\}$
 \While{$size(L) < M$}
    \State $\vec{d}_i \gets$ random unit direction
    \State $\vec{n}_i,\vec{b}_i \gets$ normal, bi-normal direction of $\vec{d}_i$ via SVD
    \State $u_0, u_1 \gets$ random offset values uniformly sampled in the interval between 0 and the bounding box diagonal length [$-\sqrt{3}$, $\sqrt{3}$]
    \State $o_i \gets u_0\vec{n}_i+u_1\vec{b}_i$
    \State $\ell_i=o_i+t\vec{d}_i$
    \If {$IntersectBoundingBox(\ell_i, bbox)$}
    \State Append $\ell_i$ to L
    \EndIf
\EndWhile
\State \textbf{return }{$L = \{\ell_i\}$}
\end{algorithmic}
\end{algorithm}
all rays that intersect the bounding box of the surface. Note that although many other alternative bounding volumes exist, e.g., a bounding sphere, we default to a non-tight bounding box following the convention in deep learning of normalizing shapes into the $[-1,1]^3$ cube, though tight bounding volumes would reduce the number of sampled rays that miss the surface and would thus be more efficient. While there are a variety of constructions, we propose to first sample a direction uniformly, and then among all rays parallel to that direction, uniformly sample ray origins so the resulting ray intersects the bounding cube. We do this using rejection sampling, by uniformly sampling a point on a $2\sqrt{3}$-sided square in the plane defined by the normal and bi-normal directions $\vec{n}_i$ and $\vec{b}_i$, respectively, until the parallel ray passing through this point intersects the bounding box. This intersection test between a ray candidate and the bounding box can be done with a fast ray-slab test, which has a low rejection rate if the bounding box is tight. See Algorithm \ref{alg:linesampling} for a more detailed version of this algorithm and Figure \ref{fig:line-sampling} for a explanatory diagram. This algorithm generalizes to any bounding primitive (e.g., spheres, k-DOPs), but bounding spheres can skip the rejection sampling step and instead directly sample ray origins on the projected disc on the $(\vec{n}_i, \vec{b}_i)$ plane.

One might wonder if a simple ray sampling algorithm would suffice, for example, picking a random point in $[-1,1]^3$ and pairing it with a random direction. However, this results in a non-uniform distribution as shown in Figure~\ref{fig:line-offsets}.

\subsection{Ray-Intersection Evaluation}

\begin{wrapfigure}[7]{L}{0.2\linewidth}
\includegraphics[width=0.12\textwidth,trim=0cm 0cm 0cm 0.4cm]{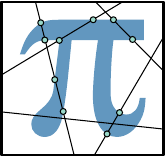}
\label{fig:intersect}
\end{wrapfigure}
Once we have $M$ randomly sampled rays, we want to find all intersections between the sampled rays and the surface (see inset, green points). 
For implicit functions, we consider a known Lipschitz bound of $\lambda$, i.e.,
$$\left|f(p_1)-f(p_2)\right|\leq\lambda \left\|p_1-p_2\right\|_2,\,  \lambda>0.$$
This means that at every query location $x$, we know we can safely take a step of size $|f(x)| / \lambda$ in any direction $\vec{d}$ without stepping over the zero level set~\cite{spheretracing}.
\begin{algorithm}[h!]
\caption{UniformPoints: Sample uniformly distributed points on the zero level set of an implicit function $f$ with Lipschitz bound $\lambda$ via random ray casting.}
\label{alg:randompoints}
\begin{algorithmic}
\Require $f$: implicit function with Lipschitz bound $\lambda$
\Require $\epsilon = 1e^{-4}$: tolerance for ray-intersection finding
\State $L \gets $ Sample $M$ rays via UniformRays() in Algorithm \ref{alg:linesampling}
\For {$i=1,2 \ldots M$}
\State $P = \{\}$
\State $t = 0$
\State $p \gets o_i+t\vec{d}_i/\lambda$
\State $s \gets |f(p)|$
\While{$t<\text{length}(\ell_i)$}
\If {$s < \epsilon$}
\State Append $p$ to $P$
\While{$s < \epsilon$}
\State $t \gets t + max(s, \epsilon)$
\State $p \gets o_i+t\vec{d}_i/\lambda$
\State $s \gets |f(p)|$
\EndWhile
\EndIf
\State $t \gets t + s$
\State $p \gets o_i+t\vec{d}_i/\lambda$
\State $s \gets |f(p)|$
\EndWhile
\EndFor
\State \textbf{return }{$P$}
\end{algorithmic}
\end{algorithm}
We slightly modify the classic sphere tracing algorithm \cite{spheretracing}, originally designed to find only the first intersection of a given ray with the level set, to instead find all intersections of a ray with the level set.
In our modification, we ray march from the origins of each ray with step size $|f(x)| / \lambda$ until we find an intersection at the zero level set when $|f(p)|<\epsilon$. After each intersection, we keep ray marching with step size $\text{max}(|f(x)|, \epsilon) / \lambda$ until the implicit function value reaches above the threshold value $\epsilon$, and proceed as usual to find the next intersection, terminating when we reach the end of the ray (i.e., the ray is outside of the box). Please see Algorithm \ref{alg:randompoints} for the detailed pseudocode. 

\begin{figure}[htb]
  \includegraphics[width=\linewidth]{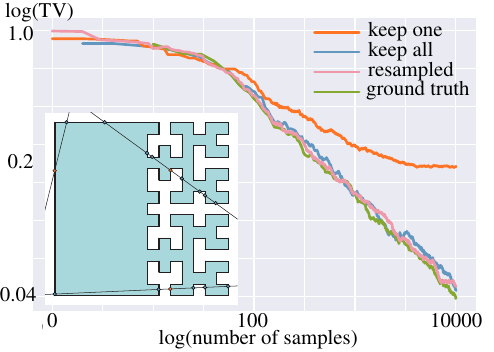}
  \caption{\label{fig:resampling}
  We plot the total variation (TV) score measuring the uniformity of sample sets acquired with uniform sampling on the ground truth polyline (``ground truth''), our method (``keep all''), the resampling procedure described in Section~\ref{sec:resampling} (``resampled'') and taking a single intersection (``keep one''). Our method and the resampling approach produce uniform samples on the surface comparable with comparable TV scores to directly sampling on the polyline, while the keeping one sample along each casted ray produces worse TV scores. 
  }
\end{figure}
As this part of our algorithm is largely based on sphere tracing, our method can benefit from any improvements proposed for sphere tracing such as \cite{segmenttracing}.

\paragraph*{Taking One Intersection}
\label{sec:alternatives}

A tempting alternative is to simply take one random intersection instead of all intersections on each casted ray. However, as shown in Figure~\ref{fig:resampling}, naively only keeping one sample on casted rays results in incorrect non-uniform samples.

\subsection{Acceleration and Stratification via Sparse Voxels}
\label{sec:sparsevoxels}

\begin{figure}
  \includegraphics[width=\linewidth]{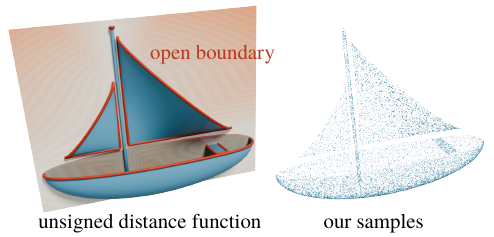}
  \caption{\label{fig:opensurface}
  Our method only requires the surface to define a codimension-one shape as the level set of an implicit function. Non-manifold junctures and open boundaries are no issue for our white noise sampling, such as this sailboat modeled as a cascade of analytic \emph{unsigned} distance functions.
  }
\end{figure}

Many neural implicit surface functions come with existing sparse voxel structures built to accelerate its optimization \cite{muller2022instant}. Our method naturally benefits from sparse voxel structures as we can divide every implicitly-defined surface into sub-surfaces with open boundaries inside each voxel, and apply our method independently within each voxel. 
Since our method only requires an implicit function of non-zero Lipschitz bound, it can work with shapes with non-manifold junctures and open boundaries as shown in Figure \ref{fig:opensurface}.

Applying our method with sparse voxel structures not only accelerates the sampling process but also provides a way to perform stratified sampling to reduce variance in Monte Carlo estimates, as shown in Figure \ref{fig:sparsevoxels}. The sparse voxel structure is also useful for input where sphere tracing-like algorithms cannot be efficiently used everywhere, e.g., many formulations of neural implicits which encourage SDF-ness but do not guarantee it and may have large global Lipschitz bounds. 

\begin{figure}
  \includegraphics[width=\linewidth]{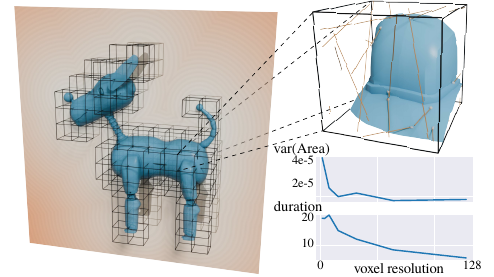}
  \caption{\label{fig:sparsevoxels}
           Sparse voxel structures can accelerate our method and act as a form of stratified sampling which reduces variance for estimations of shape quantities as described in Section \ref{sec:surface-area}. Here we show one instance of such a voxel structure with a voxel grid resolution of 16, along with a zoom-in view of one voxel, a set of sampled rays (orange) in this voxel, and the resultant samples on the zero level set (white). We plot the variance of a surface area estimate over 30 runs using sparse voxel structures built with different grid resolutions, which demonstrates the variance-reducing benefits of stratified sampling through voxels. We also plot the total time needed (in seconds) for our method across voxel resolutions using the same number of rays and show its acceleration benefits. 
           }
\end{figure}

\section{Evaluations}
An ideal uniformly distributed set of samples on the surface means the ratio of samples belonging to each local surface region should be the same as the area ratio of the local region. Therefore, if we can divide the surface into disjoint local patches, one can measure the uniformity of a sample set by computing statistical distance metrics, such as total variation (TV) distance and KL divergence, between the distribution of samples per patch and the area density of these patches. 

While in practice our method is not necessary for triangle meshes, they nevertheless provide an ideal real-world test set for measuring the behavior of our method and baselines against a ground truth sampler.
We can easily produce a ``ground truth'' uniform distribution over the mesh by sampling triangles proportional to area and uniformly sampling on each selected triangle, and we can also evaluate our method on triangle meshes through direct ray-intersection queries or sphere tracing the minimum distance to the mesh (in our experiments, we use the latter for the sake of generality).
Therefore, we focus on evaluating our method and baseline methods on signed implicit functions defined by an existing mesh dataset~\cite{fieldaligneddataset}. 

The surface defined by a triangular mesh has a natural disjoint decomposition into triangles, so for a given sample set on a mesh, we can measure its uniformity via the discrete total variation (TV) score by taking the sum over all triangles of the absolute difference between the proportion of samples in the triangle and the area-proportional density of the triangle:
\begin{equation}
\text{TV} = \frac{1}{2} \sum_i \left| \frac{n_i}{N} - \frac{A_i}{A_\text{total}} \right| 
\end{equation}
Lower TV scores indicate more uniform sample sets.

In addition to uniformity of samples, we measure the efficiency of each sampling method by counting the number of implicit function evaluation calls used during sampling. The reason we use the number of evaluations instead of run time is because many of the methods studied here can be dramatically optimized for run time due to their parallelizable nature, such as marching cubes and even our method. We believe the number of function evaluations is an accurate measurement of the method's efficiency that is agnostic to the chosen implementation, so our research prototype still provides meaningful performance statistics. 

\subsection{Warm Up: Analytic Torus}
\begin{wrapfigure}[6]{R}{0.4\linewidth}
\includegraphics[width=0.2\textwidth,trim=0.5cm 0.0cm 0.0cm 0.5cm]{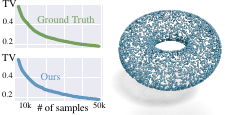}
\label{fig:torus}
\end{wrapfigure}
To start, we apply our method to an analytic implicit distance function of a torus, and compute the TV score of the generated samples. We can sample the torus uniformly via an analytical formula derived with the inverse CDF method as the ground truth.

To evaluate uniformity of our samples against the analytical samples, we parameterized the torus into a toroidal grid of resolution $100 \times 100$. Each grid patch's surface area can be analytically computed via integration. We then compute the TV distance between the surface area distribution of the toroidal patches and the distribution of number of samples per patch. See the figure on the right for a plot of TV distance against the number of samples, from both our method and the analytically-derived formula. The uniformity of samples acquired from our method on the torus's surface aligns well with the ground truth uniform sampler, verifying the theoretical uniformity of our method. 

\label{sec:evaluation}
\begin{figure}
  \includegraphics[width=\linewidth]{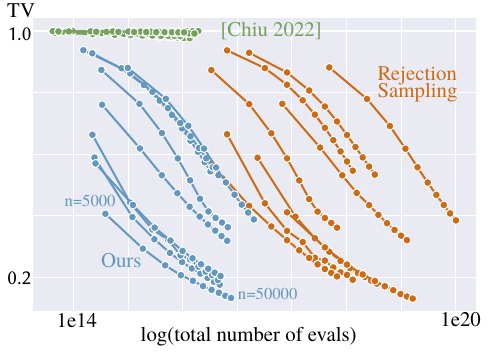}
  \caption{\label{fig:coreplot-two}
  We compare our method with the specialized Hamiltonian Monte Carlo sampling algorithm proposed in \cite{ericauniform} and rejection sampling, by drawing an increasing number of samples from 5,000 to 50,000 on a subset of 8 different implicit functions defined by meshes from a dataset \cite{fieldaligneddataset}. For each sampling run, we measure the uniformity of the resulted sample set via TV distance and the total number of function evaluations needed. Our method consistently outperforms both baseline methods in both TV distance and function evaluation count, resulting in more uniform sets of samples on surface while being less costly to evaluate. 
           }
\end{figure}

\begin{figure}
  \includegraphics[width=\linewidth]{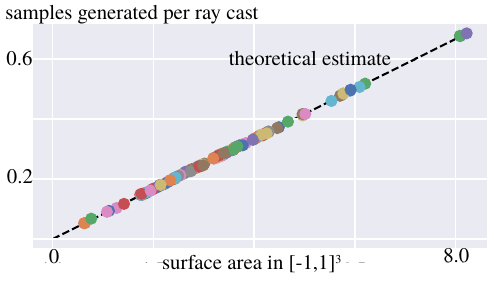}
  \caption{\label{fig:coreplot-ours}
  We ran our method on 114 different implicit surface functions from a dataset of meshes \cite{fieldaligneddataset}. Here we plot the relationship between ground truth surface area and the ratio between number of samples and the number of casted rays (i.e., average intersection count). We confirm that the the average number of samples generated per ray is linear in surface area, which corroborates the surface area formula in Section \ref{sec:surface-area}. This also means the rejection rate and wasted evaluation cost from our method correlates with the surface area, which we show in Appendix~\ref{app:eval_vs_area}. }
\end{figure}

\paragraph*{Baselines}

We next compare our methods to three baseline methods on a dataset of implicit functions defined by 114 meshes \cite{fieldaligneddataset} for which we have access to the ground truth implicit function and a ground truth uniform sampler, i.e., uniformly sampling on the mesh. For baselines, we compare with (1) uniform sampling on the extracted mesh via marching cubes, (2) sampling using a Hamiltonian Monte Carlo algorithm~\cite{ericauniform}, and (3) rejection sampling as used in many existing neural implicit techniques \cite{gpnf, diffcd}. 
\begin{figure*}[htb]
  \includegraphics[width=\linewidth]{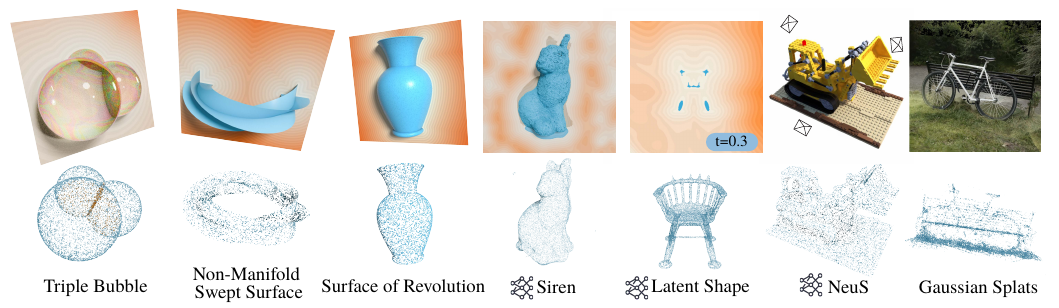}
  \caption{\label{fig:gallery}
  We apply our method to a variety of implicit functions. From left to right: a Möbius transformation of the “standard triple bubble” from \cite{triplebubble}, a swept unsigned distance function resulting in a non-manifold shape, a surface of revolution generated by a cubic B\'ezier curve, another NeuS surface trained from the Lego scene in the Blender dataset \cite{neus}, another SIREN surface fitted on LiDAR point clouds of scan55 from the DTU dataset \cite{siren}, an implicit surface from interpolating two surfaces using latent code $t=0.3$ \cite{liu2022learning}, and a preliminary result on a Gaussian Splats scene.
  }
\end{figure*}

More specifically, baseline (2) is a Markov Chain Monte Carlo method that generates a sequence of samples following a specialized density distribution proposed in the thesis of Chiu~\cite{ericauniform}, which asymptotically converges to a uniform set of samples on the surface. Implementation-wise, we mostly follow the experiment setup specified in \cite{ericauniform} except we replace the initial 1000 gradient descent steps with 5 Newton descent steps before the burn-in period. Given an initial sample $x$ randomly sampled in the bounding box and the signed distance function $f$, a Newton step is:
$$ x = x - f(x)\frac{\nabla f(x)}{||\nabla f(x)||^2}.$$ This effectively moves point $x$ to be near the surface with fewer steps than gradient descent. The rest of the hyperparameter choices stay the same: mass size is 1, total integration time is 1, the
number of integration steps is 100, and the number of burn-in Hamiltonian Monte Carlo steps is 500. Note that this specialized Hamiltonian Monte Carlo method practically can only sample near the surface with a small but non-zero $T$ parameter and therefore requires an additional projection step. Note also that both baselines of rejection sampling and sampling via marching cubes can produce samples far from the surface as well, depending on the chosen hyperparameters. We use the same Newton method for any projection step needed across all methods.

\begin{figure*}[htb]
  \includegraphics[width=\linewidth]{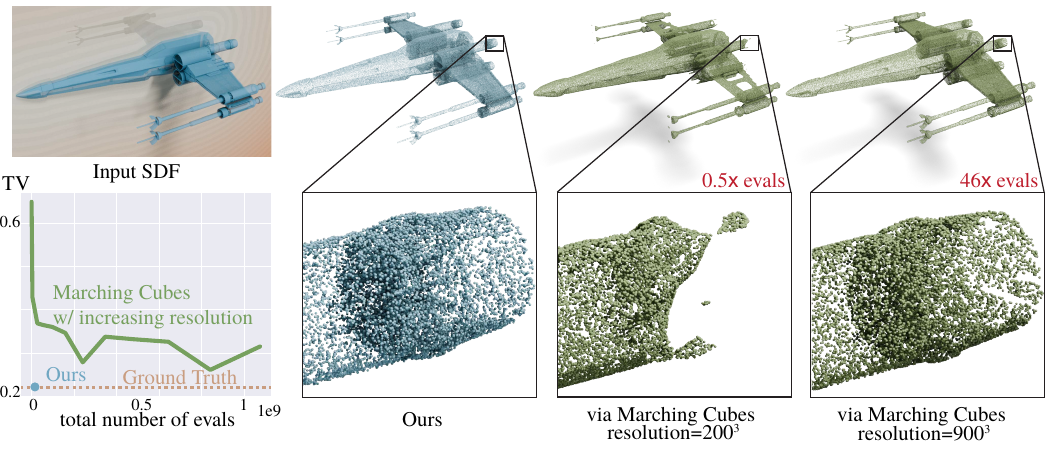}
  \caption{\label{fig:marching-cubes-comparison}
  We compare our method with sampling a mesh extracted via marching cubes, from grids with resolutions ranging from $200^3$ to $1024^3$. With this particular shape, even though the uniformity of the sample set via marching cubes improves with increasing grid resolution, the extracted surface still misses some parts of the surface with a relatively high resolution grid as shown in the zoom-in view on the right, which degrades uniformity. For sampling the same number of points, our method results in a sample set that is more uniformly distributed on the surface, measured by a lower total variation score (TV), while being much cheaper to evaluate (for reference, ours took $1.6 \times 10^7$ function evaluations). We also plot the ``Ground truth'' total variation score, averaged over 10 sampling runs using the ground truth mesh, for reference.}
\end{figure*}

\paragraph*{Results} Our method achieves the best trade-off along the axes of sample uniformity and method efficiency compared to all baselines as shown in Table \ref{tab:core-table}. 

\begin{table}[!t]
    \centering
    \caption{
        \label{tab:core-table}
        We compute the average total number of function evaluations to sample 50,000 points on a dataset of surfaces~\cite{fieldaligneddataset}, as well as the total variation (TV) score measuring the uniformity of the samples.
        Although our method is generally aimed at implicit surfaces, we use meshes here for the sake of known geometry to evaluate against.
        For marching cubes, the grid size is $1024^3$. 
        The total variation ``Ground truth'' refers to uniform sampling using the ground truth mesh, for which we use the average score across 10 sampling runs for each shape.
    }
    \resizebox{0.99\linewidth}{!}{
    \begin{tabular}{lccccc}
    \toprule
    Method & Ground Truth & \cite{ericauniform} & Rejection Sampling & via Marching Cubes & Ours \\
    \midrule
    Number of Evals & N/A & $9.98\times 10^6$ & $3.98\times 10^8$ & $1.07\times 10^9$ & $1.92\times10^7$  \\
    TV & 0.373 & 0.989 & 0.373 & 0.383 & 0.372  \\
    \bottomrule
    \end{tabular}
    }
\end{table}

In Figure \ref{fig:coreplot-two}, we compare with Hamiltonian Monte Carlo~\cite{ericauniform} and rejection sampling, which is a common method used in prior work for uniformly sampling implicit surfaces~\cite{gpnf,diffcd}. Here we sample an increasing number of samples from 20 to 50,000 on all shapes using all methods, and plot the TV distance scores and the total number of function evaluations on a randomly-chosen subset of 8 shapes. Our method is at the optimal front with respect to both metrics compared to the other baselines, i.e. consistently resulting in more uniform samples while being less expensive. As mentioned before, the theory behind the Hamiltonian Monte Carlo algorithm only guarantees uniform on-surface samples asymptotically when the temperature parameter $T$ goes to 0~\cite{ericauniform}, which is practically impossible. With a non-zero $T$, the samples usually end up floating near the surface and require an additional projection step to exactly lie on the surface. We show the samples before and after projection in Figure~\ref{fig:ficus-comparison}. In addition, the sequential sample sets acquired from such a Markov Chain Monte Carlo method often contain duplicated samples, resulting in a much worse uniformity measure.

We compare to sampling via marching cubes in a separate figure as the cost of function evaluation with marching cubes is tied with the grid resolution instead of the number of samples. As shown in Figure \ref{fig:marching-cubes-comparison}, we obtain 50,000 samples on one shape using our method and sampling on marching cubes-extracted meshes from grids of increasing resolution from $300^3$ to $1024^3$, and plot the total variation score using the ground truth mesh. Marching cubes is consistently more expensive than ours with respect to total number of queries as it grows cubically with grid resolution, and the samples are less uniform as the mesh extracted with marching cubes tend to miss thinner structures and details. While we focus on a single example in this figure, the observations generalize as demonstrated by metrics in Table~\ref{tab:core-table} computed from a dataset of shapes~\cite{fieldaligneddataset}. 

In addition to comparing with baselines, we provide additional analysis of our method itself. We focus on the relationship between the ground truth surface area and average number of samples per casted ray. As shown in Figure \ref{fig:coreplot-ours}, we observe a strong linear correlation between the average number of samples per casted ray and the surface area of the underlying implicit surface, which is what we expect from integral geometry~\cite{santalointegral}. This means the rejection rate and the amount of wasted computation from our method is closely tied with the underlying surface area as well, which we show in Appendix~\ref{app:eval_vs_area}. At the same time, this relationship allows one to use our method to estimate important shape quantities such as surface area and volume (Section~\ref{sec:surface-area}).

\section{Applications}

This section focuses on various applications of our method. Our method can be used for a range of tasks including shape quantity estimation, offset surface sampling, and curvature-based resampling. Meanwhile, extracting a uniformly distributed set of samples on the implicit surface enables a range of downstream graphics applications including blue noise sampling, neural implicit shape deformation, and many others. 

\subsection{Sampling Across Representations}
\begin{figure}
  \includegraphics[width=\linewidth]{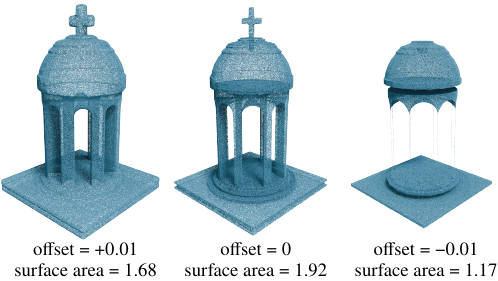}
  \caption{\label{fig:offset}
  White noise sampling can be straightforward with mesh surfaces, but offset surface sampling is often non-trivial unless converted to an implicit surface. Our method can easily sample offset surfaces, along both positive and negative directions, in addition to the surface defined at the zero level set. 
  }
\end{figure}
Our method is applicable to a variety of surface representations, including not just implicit surfaces but generally anything against which we can cast rays and gather intersections.
Figure~\ref{fig:gallery} demonstrates point samples across a collection of these representations. 
Analytical implicit surfaces can be defined by sweeping and revolution (Figure \ref{fig:gallery}), or authored by an artist (Figure \ref{fig:burger_moments}). 
Neural implicit representations (Figure \ref{fig:gallery}) such as Siren~\cite{siren}, NeuS~\cite{neus}, or latent interpolations~\cite{liu2022learning} are used in reconstruction and generative modeling.
We also include a preliminary example on a Gaussian particle surface~\cite{kerbl20233d} (Figure \ref{fig:gallery}); although there is no precise notion of a surface, our algorithm produces reasonable point sets by casting rays as in~\cite{moenne20243d}, marking an intersection when the transmittance drops below $50\%$ and casting new rays from the intersection points to obtain all intersections along the same direction.

\paragraph*{Open Boundaries and Nonmanifold Geometry}
Open surfaces, variable codimension, and non-manifold structure likewise pose no problem for our approach, so long as we can intersect rays with the surface.
Figure~\ref{fig:opensurface} shows an example sampling from an open surface by tracing an unsigned distance function, while the swept surface in Figure~\ref{fig:gallery} is non-manifold.

\paragraph*{Offset Surfaces}
Even when working with an explicit mesh representation, our approach may useful to sample from implicit, derived surfaces such as the \emph{offset surface} from a mesh, defined as the shifted surface which is a specific distance away from the input mesh.
We can easily cast rays against this offset surface by querying and shifting the distance from the mesh at any point in space; Figure \ref{fig:offset} shows one example of sampling points from such an offset surface. 
In turn this also allows the estimation of derived quantities like surface area of the offset surface, as described below.

\begin{figure}
  \includegraphics[width=\linewidth]{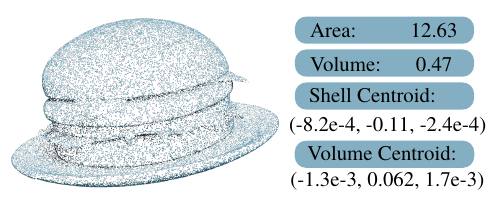}
  \caption{\label{fig:burger_moments}
  Our method can estimate various shape quantities like surface area, volume, surface  centroids, and volumetric centroids, as demonstrated here for this PseudoSDF ShaderToy example of a burger as seen in Figure \ref{fig:teaser}. ShaderToy PseudoSDF credit to \copyright Xor \href{https://creativecommons.org/licenses/by-nc-sa/3.0/}{(CC BY-NC-SA 3.0)}}
\end{figure}

\subsection{Moment Estimation}
\label{sec:moment-estimation}

Our ray-casting-based sampling can also be used to compute moments of the shape. For the calculations below we assume that the shape is enclosed in the $[-1,1]^3$ cube.
All of these formulas, or variants of them, are well-known, so we will only briefly discuss each one. See Figure \ref{fig:burger_moments} for an example of moments estimation on a pseudo-SDF ShaderToy example of a burger.

\label{sec:surface-area}
\paragraph*{Surface Area}
\begingroup
\newcommand{\numlines}[0]{M}
\newcommand{\numhits}[0]{K}
\newcommand{\chordset}[0]{\mathcal{C}}
If $\numlines$ random rays in the $[-1,1]^3$ cube against an enclosed surface $S$ results in $\numhits$ intersections, then a Cauchy-Crofton-like formula for the surface area of $S$ is:
\begin{equation}
A = \int\limits_S 1 \, d\mathbf{x} =  \lim\limits_{\numlines \rightarrow \infty} \frac{12 \numhits}{\numlines}.
\label{eq:surfacearea}
\end{equation}
This formula can be derived by taking ratios of Monte Carlo estimates from the Cauchy-Crofton formula between $S$ and $[-1,1]^3$, to eliminate the proportional constant depending on the ray integration volume~\cite{Daugmaudis2010}.

\paragraph*{Thin Shell}
For a surface $S$, averaging uniform samples directly on $S$ gives a Monte Carlo estimate of its centroid:
\begin{equation}
    \mathbf{c}_\text{shell} = \frac{1}{A} \int\limits_S \mathbf{x} \, d\mathbf{x}  = \lim\limits_{\numlines \rightarrow \infty} \ \frac{1}{\numhits} \sum\limits_{i=1}^\numhits \mathbf{x}_i.
\end{equation}

\paragraph*{Volume}

Treating the shape as enclosing a solid volumetric region $\Omega$, our intersections also give a Cauchy-Crofton-like formula for the total volume.
Similar formulas can be derived using the divergence theorem for rays in a single direction~\cite{khosravifard2010new,trusty2021shape}.
\begin{equation}
V = \int\limits_\Omega 1 \, d\mathbf{x} = \lim\limits_{\numlines \rightarrow \infty} \frac{6}{\numlines} \sum\limits_{i=1}^\numlines \sigma_i
\label{eq:volume}
\end{equation}
where $\sigma_i$ is the chord length for the $i$-th ray (i.e., the total length of the portion(s) of the ray inside the shape).
The chord length can be efficiently tracked during consecutive intersection tracing.
We can derive Eq.~\ref{eq:volume} using the fact that the average chord length of rays passing through $\Omega$ is $4V/A$~\cite{santalointegral,Mazzolo2008}, and taking the ratio of volumes between $\Omega$ and $[-1,1]^3$ to eliminate the proportional constant.

Similarly, we can approximate the volumetric centroid:
\begin{equation}
\mathbf{c}_\text{solid} = \frac{1}{V} \int\limits_{\Omega} \mathbf{x} \ d\mathbf{x} = 
\frac{
\sum\limits_{i=1}^{\numlines} \sum\limits_{j\in \chordset_i} \frac{\sigma_{ij}}{2}\left( \mathbf{a}_{ij} + \mathbf{b}_{ij}\right)
}
{
\sum\limits_{i=1}^{\numlines} \sum\limits_{j\in \chordset_i} \sigma_{ij}
}
\end{equation}
where $\chordset_i$ defines the (possibly empty) set of chords for the $i$-th ray so that $\sigma_{ij}$ is the length of the $j$-th chord of the $i$-th ray and $\mathbf{a}_{ij}$ and $\mathbf{b}_{ij}$ are that chord's corresponding start and end points.
This formula is a straightforward corollary of Eq.~\ref{eq:volume}.
\endgroup

\subsection{Curvature-based Resampling}

\begin{wrapfigure}[10]{r}{0.5\linewidth}
\includegraphics[width=0.25\textwidth,trim=0.6cm 0cm 0cm 0.3cm]{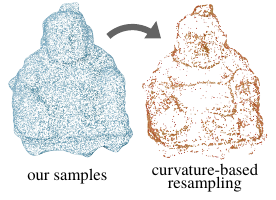}
\label{fig:curvature}
\end{wrapfigure}
With a uniformly distributed set of samples that covers every part of the surface, one can then estimate per-point metrics such as curvature and loss values for importance sampling or hard sample mining. 

Such techniques have shown benefits for tasks like neural implicit surface optimization from images \cite{isopoints}. On the right we show points obtained via curvature-based resampling on the intermediate surface during optimization \cite{siren}.  

\subsection{Low-Discrepancy Sequence Sampling}
\label{sec:bluenoise}
One direct downstream application that requires a uniform set of samples is blue noise sampling via subsampling \cite{bluenoise}. These samples can then be used for texture synthesis, simulation or remeshing. See Figure \ref{fig:bluenoise} for an example, where we subsample our uniform a samples as blue noise.

One can also replace the uniform ray sampling specified in Section \ref{sec:ray-sampling} with low-discrepancy sequences of rays as proposed in \cite{liu2006quasi, liu2010surface}.
This does not necessarily yield a corresponding low-discrepancy sequence on the surface, however we find that it indeed improves convergence when estimating moments, as shown in Figure \ref{fig:uniform_vs_lds}.
\begin{figure}
    \includegraphics[width=\linewidth, trim=0cm 0cm 0cm 0cm]{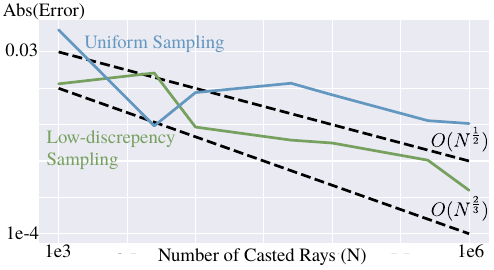}
  \caption{
  \label{fig:uniform_vs_lds}
  We plot the absolute surface area estimation error with samples acquired from sampled rays using a low-discrepancy sequence as proposed in \cite{liu2006quasi, liu2010surface}, as well as uniformly sampled rays, as described in Section \ref{sec:ray-sampling}. We also plot asymptotes $O(N^{\frac{2}{3}})$ and $O(N^{\frac{1}{2}})$, and observe that using the low-discrepancy sequence results in faster convergence, as pointed out in \cite{liu2006quasi}.} 
\end{figure}

\subsection{Neural Implicit Deformation}
\label{sec:deformation}

\begin{wrapfigure}[9]{r}{0.4\linewidth}
\includegraphics[width=\linewidth,trim=0.6cm 0cm 0cm 0.5cm]{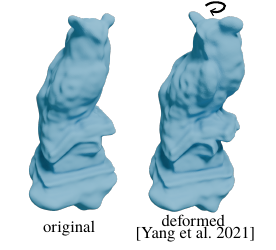}
\label{fig:deformation}
\end{wrapfigure}

A set of uniformly distributed samples on a surface is also essential for defining surface loss functions for neural implicit deformation as described in \cite{gpnf}. The original paper considers a Langevin dynamics-based sampling method which suffers from clumping near high-curvature regions and empirically uses rejection sampling with projection in their method, which we compared to as one of the baselines in Section \ref{sec:evaluation}. Here we replace the sampling module in the original implementation with our method and show one result on the right. 
\subsection{Two Dimensional Shapes}

\setlength{\columnsep}{1em}
\setlength{\intextsep}{0em}
\begin{wrapfigure}[9]{r}{0.4\linewidth}
\includegraphics[width=\linewidth]{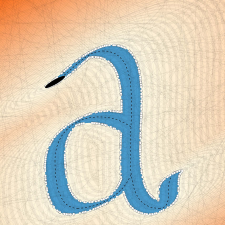}
\label{fig:metafont}
\end{wrapfigure} 
While our primary motivation is sampling surfaces in 3D, our algorithm trivially generalizes to 2D.
The recent surge of interest in generative vector graphics \cite{yang2025omnisvgunifiedscalablevector} and diffusion models for vector fonts \cite{ThamizharasanLA24} motivates the need for sampling structured shape outlines in the plane. There is already success utilizing implicit functions for these tasks \cite{ReddyZWFJM21}. The inset figure shows a classic Metafont `a' --- defined as the Minkowski sum of a rotated ellipse along a spline \cite{KnuthMFbook} --- treated as an implicit function, whose boundary is uniformly sampled with our approach restricted to 2D, zoom to see samples.
Generalization to 4D and beyond is also exciting to consider, though the average chord length formulas~\cite{santalointegral,Mazzolo2008} indicate decreasing efficiency with dimension.

\section{Conclusions}

This work considers the problem of uniformly randomly sampling surfaces without an explicit mesh representation, such as implicit surfaces. We presented a solution via randomly sampled rays cast in space and finding ray intersections with the implicit level set. While our method is grounded in existing theory and prior explorations, we evaluate it on a wide variety of implicit surfaces including modern neural representations. By leveraging the power of modified sphere tracing, we show its advantages over common baselines such as sampling on an extracted mesh, rejection sampling, and a Markov Chain Monte Carlo-based method. We explore many downstream applications that extend from our method or benefit from a uniformly-distributed set of samples on a surface. 

\paragraph*{Limitations and Future Work}

For shapes with extremely thin sparse features which occupy only a small fraction of their bounding box, sampling may become inefficient as most rays do not intersect the geometry.
All sampling methods struggle with this case in some form or another: grid-based approaches such as marching cubes may be entirely alias and miss these features, while MCMC-type methods (e.g., \cite{ericauniform}) may see proposal rejections increase leading to undersampling and over-duplicating.
We believe the behavior of our method is often preferential in this case, in that its sampling is still accurate but merely less efficient.
There are many avenues to improve the efficiency of our method as well, by using tighter bounding primitives and/or by using bounding primitives that support faster ray sampling, such as spheres.

Many of our examples leverage sphere tracing of SDF-like implicit functions; this fast tracing algorithm is only possible when the function has a known (or estimable) Lipschitz constant,  otherwise this method must fall back on dense marching or specialized searches (e.g., \cite{gillespie2024ray}).
Other common implicits such as truncated signed distance fields or neural fields \cite{xiesurvey} that are constructed with accompanying space-skipping data-structures are readily incorporated into our approach.

Some implicit formulations might model hard surfaces as a $0$-$1$ binary occupancy field or a similar density formulation; such fields are difficult to efficiently trace and sample from, as there is little information in the vicinity of the interface.
The stochastic framework for Poisson surface reconstruction \cite{sellan2022} provides one promising starting point, modeling the smoothed near-surface uncertainty which is inherently present.

There are still many interesting use cases that we would like to explore as future work. For example, certain simulators require estimation of intersection volumes between two shapes \cite{tong2018computation}, for which our shape quantity estimation property can be helpful. 
One could also go one dimension higher, and instead of sampling points via random ray casting, sample curves via random plane intersections as studied in \cite{breiding2020random}.

\section*{Acknowledgments}
We thank Derek Liu and Silvia Sellán for providing codes. Our research is funded in part by NSERC
Discovery (RGPIN–2022–04680), the Ontario Early Research Award
program, the Canada Research Chairs Program, a Sloan Research
Fellowship, the DSI Catalyst Grant program and gifts by Adobe Inc.

\bibliographystyle{eg-alpha-doi} 
\bibliography{uniform_sampling}

\appendix
\section{Proof of Resampling Formulation}\label{app:resampling}
Here we show that the ray resampling method described in Section~\ref{sec:resampling} produces uniformly distributed samples over a surface $S$.
Before stating the theorem, we define $\mathbf{L}$ as a compact set of rays, which contains all rays in $\mathbb{R}^3$ that intersect $S$.
Now, we wish to prove the following:
\begin{theorem}
Given a set of uniformly $M$ distributed rays $L = \{ \ell_i \}$ sampled from $\mathbf{L}$, where each $\ell_i$ intersects $S$ $k_i$ times, resampling $N$ rays from $L$ with probability proportional to $k_i$ and uniformly selecting one intersection per resampled ray produces a set of points $X = \{ x_j \}$ that is uniformly distributed over $S$.
\end{theorem}

\begin{proof}
We will prove the theorem by showing that $X$ is 1-equidistributed; that is, we can build a consistent Monte Carlo estimator of integrals over $S$ using $X$ and uniform contribution weights.
We resample rays by following Resampled Importance Sampling (RIS) \cite{talbot2005importance}, using a target function of $q(\ell_i) = k_i$ and noting that the base distribution is $p = p(\ell_i) = 1/| \mathbf{L} |$.
We assign a resampling weight $w_i = \frac{q_i}{p}$ to each $\ell_i$, and resample $N$ rays proportional to $w_i$.
From these resampled rays $L_r = \{ \ell_j \}$, we uniformly select an intersection with $S$, $\mathbf{x}_j$, with probability $1/k_j$.
Through an extension of the Cauchy-Crofton formula for evaluating integrals over $S$~\cite{palaispoints}, we know that in $\mathbb{R}^3$, $\int_S f(x) dx = \frac{1}{2\pi} \int_\mathbf{L} \sum_{i=1}^{k(\ell)}f(x_i) d\ell$, i.e., an integral over surfaces can be converted to an integral over rays by summing the integrand over each of the $k(\ell)$ intersections $\ell$ makes with $S$, denoted $x_i$.
Using RIS, the Monte Carlo estimator for an arbitrary function $f$ over $S$ is then
\begin{align*}
    I &= \frac{1}{2\pi}\frac{1}{N} \sum_j \frac{1}{1/k_j} \frac{f(x_j)}{q(\ell_j)} \left( \frac{1}{M} \sum_i w_i \right)  \\
    &= \frac{1}{2\pi}\frac{1}{N} \sum_j f(x_j) | \mathbf{L} | \frac{\sum_i k_i}{M}
\end{align*}
From \cite{palaispoints}, we know that as $M \to \infty$, $\frac{\sum_i k_i}{M} \to \frac{2\pi |S|}{| \mathbf{L} |}$\footnote{The proof of this result in \cite{palaispoints} has an error where they omit a factor of $\frac{2}{|  \mathbf{L} |}$, though a corrected proof follows essentially the same argument presented in the paper, by applying the Cauchy-Crofton formula and accounting for $| \mathbf{L} |$ in Monte Carlo estimates over $\mathbf{L}$.}, so the equation simplifies to $I = \frac{1}{N} \sum_j |S| f(x_i)$ as $M \to \infty$, which is the desired result.
\end{proof}
\begin{figure}
\includegraphics[width=\linewidth, trim=0cm 0cm 0cm 0.2cm]{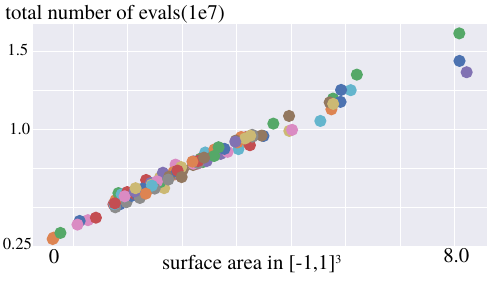}
  \vspace{-7mm}
  \caption[trim=0cm 0cm 0cm 1cm]{\label{fig:eval_vs_area}
  We ran our method to sample points with 50,000 rays on 114 shapes from a dataset of meshes~\cite{fieldaligneddataset}, and plotted the total number of evaluations required and the ground truth surface area. 
  }
\end{figure}

\section{Analysis: Evaluation Count vs. Surface Area}
\label{app:eval_vs_area}
Besides analyzing the average number of samples per casted ray, we also empirically analyzed the relationship between the total function evaluation cost of our method and the surface area of the shape being sampled. Not surprisingly, as shown in Figure~\ref{fig:eval_vs_area}, it also has an approximately linear relationship as sphere tracing-like methods are most costly near the surface.

\end{document}